\documentclass[11pt]{article}

\usepackage{fullpage,times,latexsym,amsmath,amsfonts}

\def\01{\{0,1\}}
\newcommand{\ceil}[1]{\lceil{#1}\rceil}

\newcommand{\eps}{\varepsilon}
\newcommand{\rdeg}{\mathrm{rdeg}}
\newcommand{\sign}{\mathrm{sgn}}
\newcommand{\PostQ}{\mathrm{PostQ}}
\newcommand{\ket}[1]{|#1\rangle}
\newcommand{\E}{\mathbb{E}}
\newcommand{\bra}[1]{\langle#1|}

\newcommand{\inp}[2]{\langle{#1}|{#2}\rangle} 

\newcommand{\MAJ}{\mbox{\rm MAJ}}
\newcommand{\maj}{\mbox{\rm maj}}

\newtheorem{theorem}{Theorem}
\newtheorem{lemma}{Lemma}

\newtheorem{corollary}{Corollary}

\newenvironment{proof}
{\noindent {\bf Proof. }}
{{\hfill $\Box$}\\
\smallskip}

\begin{document}

\title{Rational approximations and quantum algorithms with postselection}
\author{Urmila Mahadev\thanks{University of California, Berkeley, urmilamahadev@gmail.com.}
\and
Ronald de Wolf\thanks{CWI and University of Amsterdam, rdewolf@cwi.nl. Partially supported by a Vidi grant from the Netherlands Organization for Scientific Research (NWO), ERC Consolidator Grant QPROGRESS, and the European Commission IST STREP project Quantum Algorithms (QALGO) 600700.}
}
\date{}
\maketitle

\begin{abstract}
We study the close connection between rational functions that approximate a given Boolean function, and quantum algorithms that compute the same function using postselection. We show that the minimal degree of the former equals (up to a factor of~2) the minimal query complexity of the latter.
We give optimal (up to constant factors) quantum algorithms with postselection for the Majority function, slightly improving upon an earlier algorithm of Aaronson. Finally we show how Newman's classic theorem about low-degree rational approximation of the absolute-value function follows from these algorithms.
\end{abstract}

\section{Introduction}

\subsection{Background: low-degree approximations from efficient quantum algorithms}

Since the introduction of quantum computing in the 1980s~\cite{feynman:simulating,deutsch:uqc}, most research in this area has focused on trying to find applications where quantum computers significantly outperform their classical counterparts: new quantum algorithms, quantum cryptography, communication schemes, uses of entanglement etc. One of the more surprising applications of quantum computing in the last decade has been its use, in some way or other, in obtaining results in \emph{classical} computer science and mathematics (see~\cite{drucker&wolf:qproofs} for a survey). One direction here has been the use of quantum query algorithms to show the existence of low-degree polynomial approximations to various functions.  This direction started with the observation~\cite{fortnow&rogers:limitations,bbcmw:polynomialsj} that the acceptance probability of a $T$-query quantum algorithm with $N$-bit input can be written as an $N$-variate multilinear polynomial of degree at most $2T$. For example, Grover's $O(\sqrt{N})$-query algorithm for finding a~1 in an $N$-bit input~\cite{grover:search} implies the existence of an $N$-variate degree-$O(\sqrt{N})$ polynomial that approximates the $N$-bit OR-function, and (by symmetrization) of a univariate polynomial $p$ such that $p(0)=0$ and $p(i)\approx 1$ for all $i\in\{1,\ldots,N\}$.  Accordingly, one way to design (or prove the existence of) a low-degree polynomial with a certain desired behavior, is to design an efficient quantum algorithm whose acceptance probability has that desired behavior.  Results based on this approach include tight bounds on the degree of low-error approximations for symmetric functions~\cite{wolf:degreesymmf}, a new quantum-based proof of Jackson's theorem from approximation theory~\cite{drucker&wolf:jackson}, and tight upper bounds for sign-approximations of formulas~\cite{lee:signdegree}.

\subsection{Quantum algorithms with postselection}

In this paper we focus on a related but slightly more complicated connection, namely the use of quantum query algorithms \emph{with postselection} to show the existence of low-degree \emph{rational} approximations to various functions. We will define both terms in more detail later, but for now let us just state that postselection is the (physically unrealistic) ability of an algorithm to choose the outcome of a measurement, thus forcing a collapse of the state to the corresponding subspace. Postselection allows some functions to be computed much more efficiently. A good example of this is the $N$-bit OR function, which takes value~1 if the input $x\in\01^N$ contains at least one~1, and takes value 0 otherwise. Grover's algorithm takes $O(\sqrt{N})$ queries to compute this, which is known to be optimal (precise understanding of this algorithm and its optimality are not required for this paper).
However, a postselection algorithm could choose a tiny but positive $\eps$ and start with initial state
$$
\eps\ket{0}\ket{1}+\sqrt{\frac{1-\eps^2}{N}}\sum_{i=1}^N\ket{i}\ket{0}.
$$
Making one quantum query to the input gives
$$
\eps\ket{0}\ket{1}+\sqrt{\frac{1-\eps^2}{N}}\sum_{i=1}^N\ket{i}\ket{x_i}.
$$
Now postselect on the last qubit having value~1. This collapses the state to
$$
\eps\ket{0}\ket{1}+\sqrt{\frac{1-\eps^2}{N}}\sum_{i:x_i=1}\ket{i}\ket{1},
$$
times a normalizing constant $1/\sqrt{\eps^2+|x|(1-\eps^2)/N}$.
If $x=0^N$ then the state is simply $\ket{0}\ket{1}$, and measuring the first register gives outcome~0 with certainty. If $x\neq 0^N$, then (assuming $\eps^2\ll 1/N$) measuring the first register will probably give an index~$i$ for which $x_i=1$.  Thus we can compute OR using only one query. The error probability can be made arbitrarily small (though not~0!) by choosing $\eps$ to be very small.

\subsection{Rational functions}

A rational function is the ratio of two polynomials. Its degree is the maximum of the degrees of the numerator and denominator polynomials.  For example, here is a degree-1 rational approximation to OR (again fix small $\eps>0$):
$$
\frac{\sum_{i=1}^N x_i}{\eps+\sum_{i=1}^N x_i}.
$$
This rational function equals 0 if $x=0^N$, and equals essentially~1 if $x\neq 0^N$.  Thus it approximates the OR function very well, using only degree-1 numerator and denominator. Again, the error can be made arbitrarily small (though not~0!) by choosing $\eps$ to be very small. In contrast, a polynomial that approximates OR up to constant error needs degree $\Theta(\sqrt{N})$\cite{nisan&szegedy:degree}.

It is no coincidence that for the OR function both the complexity of postselection algorithms and the rational degree are small. The connection between postselection and rational approximation was first made by Aaronson.  In~\cite{aaronson:pp}, he provided a new proof of the breakthrough result of Beigel et al.~\cite{brs:pp} that the complexity class PP is closed under intersection. He did this in three steps:
\begin{enumerate}
\item Define a new class PostBQP, corresponding to polynomial-time quantum algorithms augmented with postselection.
\item Prove that PP = PostBQP.
\item Observe that PostBQP is closed under intersection, which is obvious from its definition.
\end{enumerate}
While very different from the proof of Beigel et al.\ (at least on the surface), Aaronson noted that his proof could actually be viewed as implicitly constructing certain low-degree rational approximations to the Majority function\footnote{The $N$-bit Majority is the Boolean function defined by $\MAJ_N(x)=1$ iff the Hamming weight $|x|:=\sum_{i=1}^N x_i$ is $\geq N/2$.}; the fact that the resulting polynomial has low degree follows from the fact that Aaronson's algorithm makes only few queries to the input of Majority.  Such rational approximations also form the key to the proof of Beigel et al.

Our goal in this paper is to work out this connection between rational functions and postselection algorithms in much more detail, and to apply it elsewhere.

\subsection{Definitions}

In order to be able to state our results, let us be a bit more precise about definitions. 

\paragraph{Polynomial approximation.}
An $N$-variate polynomial is a function $P:S^N\rightarrow\mathbb{R}$ that can be written as $P(x_1,\ldots,x_N)=\sum_{d_1,\ldots,d_N}c_{d_1,\ldots,d_N}\prod_{i=1}^N x_i^{d_i}$ with real coefficients $c_{d_1,\ldots,d_N}$. In our applications, the domain $S$ of each input variable will be either $\mathbb{R}$ or $\01$.  The \emph{degree} of $P$ is $\deg(P)=\max\{\sum_{i=1}^N d_i\mid c_{d_1,\ldots,d_N}\neq 0\}$.  When we only care about the behavior of the polynomial on the Boolean cube $\01^N$, then $x_i^d=x_i$ for all $d\geq 1$, so then we can restrict to \emph{multilinear} polynomials, where the degree in each variable is at most~1 (and the overall degree is at most~$N$).  Let $\eps\in[0,1/2)$ be some fixed constant. A polynomial $P$ \emph{$\eps$-approximates} $f:S^N\rightarrow\mathbb{R}$ if $|P(x)-f(x)|\leq\eps$ for all $x\in S^N$. The $\eps$-approximate degree of $f$ (abbreviated $\deg_\eps(f)$) is the minimal degree among all such polynomials $P$. The \emph{exact} degree of~$f$ is~$\deg(f)=\deg_0(f)$.

\paragraph{Rational approximation.}
A \emph{rational function} is a ratio $P/Q$ of two $N$-variate polynomials $P,Q:S^N\rightarrow\mathbb{R}$, where $Q$ is required to be nonzero everywhere on $S^N$ to prevent division by~0. Its degree is the maximum of the degrees of $P$ and $Q$.
A rational function $P/Q$ \emph{$\eps$-approximates} $f$ if $|P(x)/Q(x)-f(x)|\leq\eps$ for all $x\in S^N$. The $\eps$-approximate rational degree of $f$ (abbreviated $\rdeg_\eps(f)$) is the minimal degree among all such rational functions. The \emph{exact} rational degree of~$f$ is~$\rdeg_0(f)$.

\paragraph{Quantum query algorithms with postselection.}
A quantum query algorithm \emph{with postselection} (short: postselection algorithm) is a regular quantum query algorithm~\cite{buhrman&wolf:dectreesurvey} with two output bits $a,b\in\01$. We say the postselection algorithm computes a Boolean function $f:\01^N\rightarrow\01$ with error probability~$\eps$ if for every $x\in\01^N$, we have $\Pr[a=1]>0$ and $\Pr[b=f(x)\mid a=1]\geq 1-\eps$.  The idea is that we can compute $f(x)$ with error probability~$\eps$ if we could postselect on measurement outcome $a=1$. In other words, the second output bit~$b$ computes the function when the first is forced to output~1.  This ``forcing'' is the postselection step, which is not something we can actually implement physically; in that respect the model of postselection is mostly a tool for theoretical analysis, not a viable model of computation.  The \emph{postselection query complexity} $\PostQ_\eps(f)$ of $f$ is the minimal query complexity among such algorithms.\footnote{The way we defined it here, a postselection algorithm involves only one postselection-step, namely selecting the value~$a=1$. However, we can also allow intermediate postselection steps without changing the power of this model, see~\cite[Section~4.3]{drucker&wolf:qproofs}.}

\subsection{Our results}\label{ssecourresults}

\paragraph{Rational degree $\approx$ quantum query complexity with postselection.}
Our first result in this paper (Section~\ref{secrdeg=qpost}) is to give a very tight connection between rational approximations of a Boolean function $f:\01^N\rightarrow\01$ and postselection algorithms computing $f$ with small error probability. We show that the minimal degree needed for the former equals the minimal query complexity needed for the latter, to within a factor of~2:
$$
\frac{1}{2}\rdeg_\eps(f)\leq \PostQ_\eps(f) \leq \rdeg_\eps(f).
$$
In other words, minimal rational degree is essentially equal to quantum query complexity with postselection.  The fact that low query complexity of postselection algorithms gives low rational degree has been known since Aaronson's paper~\cite{aaronson:pp}; what we add in this paper is the converse, that low rational degree also gives efficient postselection algorithms. This tight relation (to within a factor of~2) should be contrasted with the better-studied case of polynomial approximation, where the approximate degree $\deg_\eps(f)$ equals the bounded-error quantum query complexity \emph{to within a polynomial factor}~\cite{bbcmw:polynomialsj}, and there are actually polynomial gaps~\cite{ambainis:degreevsquery}.

\paragraph{Optimal postselection algorithm for Majority.}
In his paper, Aaronson~\cite[Theorem~4]{aaronson:pp} implicitly gave an efficient postselection algorithm for the Majority function with polynomially small error probability:
$$
\PostQ_{1/N}(\MAJ_N)=O\left( (\log N)^2 \right).
$$
For constant error probability, one can obtain a postselection algorithm using $O(\log(N)\log\log(N))$ queries from his proof~\cite[Theorem~4.5]{drucker&wolf:qproofs}.

Our second result in this paper is to optimize Aaronson's construction to have minimal query complexity up to a constant factor (and hence the induced rational approximation for majority will have minimal degree), for every error probability $\eps\in(2^{-N},1/2)$:
$$
\PostQ_\eps(\MAJ_N)=O\left( \log(N/\log(1/\eps))\log(1/\eps) \right).
$$
Combined with the above constant-factor equivalence of $\rdeg_\eps(f)$ and $\PostQ_\eps(f)$, this reproves the upper bound of Sherstov~\cite[Theorem~1.7]{sherstov:intersection}.  In fact, we could just have combined Sherstov's upper bound with that equivalence, but our derivation of minimal-degree polynomials by means of a postselection algorithm is very different from Sherstov's proof. Sherstov's matching lower bound for the degree of rational approximations shows that also our algorithm is optimal (up to a constant factor).

\paragraph{Newman's Theorem.}
One of the most celebrated results in rational approximation theory is Newman's Theorem~\cite{newman:approx}. This says that there is a degree-$d$ rational function that approximates the absolute-value function $|x|$ on the interval $x\in[-1,1]$ up to error $2^{-\Omega(\sqrt{d})}$. In contrast, it can be shown that the smallest error achievable by degree-$d$ polynomials is $\Theta(1/d)$. The proof of Newman's Theorem is not extremely complicated:
\begin{quote}
Define $a=e^{-1/\sqrt{d}}$, $p(x)=\prod_{k=0}^{d-1}(a^k + x)$, and degree-$d$ rational function $r(x)=\frac{p(x)-p(-x)}{p(x)+p(-x)}$.\\
Half a page of calculations shows that $r(x)$ $\eps$-approximates the sign-function on the interval $[-1,-\eps]\cup[\eps,1]$, for $\eps=e^{-\Omega(\sqrt{d})}$.  We have $r(x)\in[-1,1]$ and $\sign(x)=\sign(r(x))$ on the whole interval $[-1,1]$, hence the degree-$(d+1)$ rational function $x\cdot r(x)$ $\eps$-approximates the absolute-value function on the whole interval $[-1,1]$.
\end{quote}
In fact the optimal error $\eps$ achievable by degree-$d$ rational functions is known much more precisely~\cite[Theorem~4.2]{petrushev&popov:rational}: it is $\Theta(e^{-\pi\sqrt{d}})$.  The proof of this tighter bound is substantially more complicated.\footnote{In fact, in the 19th century Zolotarev~\cite{zolotarev} already gave the optimal polynomial for each degree~$d$.  Later, Akhiezer~\cite{akhiezer:zolotarev} worked out the asymptotic decrease of the error as a function of~$d$, stating Newman's Theorem much before the paper of Newman (who was apparently unaware of this Russian literature).}

In Section~\ref{secnewman} we show how our postselection algorithm for Majority can be used to derive Newman's Theorem.\footnote{Actually, Aaronson's above-mentioned $O((\log N)^2)$-query postselection algorithm with error $\eps=1/N$ can already be used for this purpose; this application does not require our optimized version of the algorithm.}  While this proof is not easier than Newman's by any reasonable standard, it (like the reproof of Sherstov's result mentioned above) is still interesting because it gives a new, quantum-algorithmic perspective on these known results that may have other applications.

\section{Query complexity with postselection $\approx$ degree of rational approximation}\label{secrdeg=qpost}


We first show that rational approximation degree and quantum query complexity with postselection are essentially the same for all Boolean functions.

\begin{theorem}\label{thrdeg2postq}
For all $\eps\in[0,1/2)$ and $f:\01^N\rightarrow\01$ we have
$\rdeg_\eps(f)\leq 2\PostQ_\eps(f)$.
\end{theorem}

\begin{proof}
Consider a postselection algorithm for $f$ with $T=\PostQ_\eps(f)$ queries and error $\eps$. Then by~\cite{bbcmw:polynomialsj},  the probabilities $Q(x)=\Pr[a=1]$ and $P(x)=\Pr[a=b=1]$  can be written as polynomials of degree $\leq 2T$. Their ratio $P/Q$ is a rational function that equals the conditional probability $\Pr[b=1\mid a=1]$. By definition, the latter is in $[1-\eps,1]$ for inputs $x\in f^{-1}(1)$, and is in $[0,\eps]$ for $x\in f^{-1}(0)$. Hence $P/Q$ is a rational function of degree $\leq 2T=2\PostQ_\eps(f)$ that $\eps$-approximates~$f$.
\end{proof}

\begin{theorem}
For all $\eps\in[0,1/2)$ and $f:\01^N\rightarrow\01$ we have
$\PostQ_\eps(f)\leq \rdeg_\eps(f)$.
\end{theorem}

\begin{proof}
Consider a rational function $P/Q$ of degree $d=\rdeg_\eps(f)$ that $\eps$-approximates~$f$. It will be convenient to convert $f$ to a $\pm 1$-valued function. Define $F(x)=1-2f(x)\in\{\pm 1\}$ and $R(x)=Q(x)-2P(x)$, then $R/Q=1-2P/Q$ is in $[-1-2\eps,-1+2\eps]$ if $F(x)=-1$, and in $[1-2\eps,1+2\eps]$ if $F(x)=1$. We will write $R$ and $Q$ in their \emph{Fourier decompositions}:\footnote{The \emph{Fourier coefficients} of a function $g:\01^N\rightarrow\mathbb{R}$ are $\widehat{g}(S)=\frac{1}{2^N}\sum_{x\in\01^N}g(x)(-1)^{x\cdot S}$, where $S\in\01^n$ corresponds to a subset of $[N]$ (i.e., a subset of the $N$ input variables); $x\cdot S$ denotes the inner product between the two $N$-bit strings $x$ and $S$. The Fourier decomposition of $g$ is $g(x)=\sum_{S}\widehat{g}(S)(-1)^{x\cdot S}$.} 
$$
R(x)=\sum_{S\subseteq[N]}\widehat{R}(S)(-1)^{x\cdot S}\mbox{ \ and \ }Q(x)=\sum_{S\subseteq[N]}\widehat{Q}(S)(-1)^{x\cdot S}.
$$
Now set up the following $(N+1)$-qubit state (up to a global normalizing constant):
$$
\ket{0}\sum_S\widehat{Q}(S)\ket{S} + \ket{1}\sum_S\widehat{R}(S)\ket{S},
$$
where $\ket{S}$ is the $N$-bit basis state corresponding to the characteristic vector of~$S$.
Note that $\widehat{R}(S)$ and $\widehat{Q}(S)$ are~0 whenever $|S|>d$.
Hence by making $d$ queries to~$x$, successively querying the indices $i\in S$ and adding their value as a phase~$(-1)^{x_i}$, we can add the phases $(-1)^{x\cdot S}$:
$$
\ket{0}\sum_S\widehat{Q}(S)(-1)^{x\cdot S}\ket{S} + \ket{1}\sum_S\widehat{R}(S)(-1)^{x\cdot S}\ket{S}.
$$
Now a Hadamard transform on each of the $n$ qubits of the second register gives a state proportional to
\begin{align*}
\ket{0}\left(\sum_S\widehat{Q}(S)(-1)^{x\cdot S}\ket{0^N}+\cdots\right) + \ket{1}\left(\sum_S\widehat{R}(S)(-1)^{x\cdot S}\ket{0^N}+\cdots\right) \\
 = \ket{0}\left(Q(x)\ket{0^N}+\cdots\right) + \ket{1}\left(R(x)\ket{0^N}+\cdots\right),
\end{align*}
where the $\cdots$ indicates all the basis states other than $\ket{0^N}$.
Postselect on measuring $\ket{0^N}$ in the second register (more precisely, set the bit~$a$ to~1 only for basis state~$\ket{0^N}$).  What is left in the first register is the following qubit:
$$
\ket{\beta_x}=c(Q(x)\ket{0}+R(x)\ket{1})=cQ(x)\left(\ket{0}+\frac{R(x)}{Q(x)}\ket{1}\right),
$$
where $c=1/\sqrt{Q(x)^2+R(x)^2}$ is a normalizing constant. Since $R(x)/Q(x)\approx F(x)\in\{\pm 1\}$, a Hadamard transform followed by a measurement will with high probability tell us the sign $F(x)$ of $R(x)/Q(x)$. If $F(x)=1$, the error probability equals
$$
|\inp{-}{\beta_x}|^2=\frac{(Q(x)-R(x))^2}{2(Q(x)^2+R(x)^2)}=\frac{(1-R(x)/Q(x))^2}{2(1+(R(x)/Q(x))^2)}\leq\frac{(2\eps)^2}{2(1+(1-2\eps)^2)}=\frac{\eps^2}{1-2\eps+2\eps^2}\leq\eps,
$$
where the last inequality used that $\eps\leq 1-2\eps+2\eps^2$ for all $\eps\in[0,1/2)$. If $F(x)=-1$ then an analogous calculation works. Hence we have found a $d$-query postselection algorithm that computes $f$ with error probability $\leq \eps$.
\end{proof}


\section{An optimal postselection algorithm for Majority}\label{secoptalgomaj}

In this section we give an optimized postselection algorithm for Majority, slightly improving Aaronson's construction.  We will require the following result from~\cite[first paragraphs of proof of Theorem~4]{aaronson:pp}:

\begin{lemma}[Aaronson]
Let $\alpha,\beta>0$ satisfy $\alpha^2+\beta^2=1$. Using one query to input $x\in\01^N$ and postselection, we can construct the following qubit:
\begin{equation}\label{eqpostqubit}
c\left(\alpha |x|\ket{0} + \beta\frac{N - 2|x|}{\sqrt{2}}\ket{1}\right),
\end{equation}
where $c=1/\sqrt{\alpha^2|x|^2 + \frac{\beta^2}{2}(N - 2|x|)^2}$ is a normalizing constant.
\end{lemma}

For the sake of being self-contained, we repeat Aaronson's proof below.

\medskip

\begin{proof}
Assume for simplicity that $N$ is a power of~2, so $N=2^n$ and we can identify the indices $i\in[N]$ with $n$-bit strings. Let $s=|x|$.
Start with $(n+1)$-qubit state~$\ket{0^{n+1}}$, and apply Hadamard transforms to the first $n$ qubits and then one query to~$x$, to obtain
$$
\frac{1}{\sqrt{N}}\sum_{i\in\01^n}\ket{i}\ket{x_i}.
$$
Again apply Hadamard transforms to the first $n$ qubits, and postselect on the first $n$ qubits being all-0. Up to a normalizing constant, the last qubit will now be in state
$$
\ket{\psi}=(N-s)\ket{0}+s\ket{1}.
$$
Add a new qubit prepared in state $\alpha\ket{0}+\beta\ket{1}$ to (the left of) this qubit~$\ket{\psi}$. Conditioned on this new qubit, apply a Hadamard transform to $\ket{\psi}$, giving
\begin{eqnarray*}
\alpha\ket{0}\ket{\psi}+\beta\ket{1}H\ket{\psi} & = & \alpha\ket{0}\left((N-s)\ket{0}+s\ket{1}\right) + \beta\ket{1}\left(\frac{N}{\sqrt{2}}\ket{0}+\frac{N-2s}{\sqrt{2}}\ket{1}\right)\\
 & = & \left(\alpha (N-s)\ket{0}+\beta\frac{N}{\sqrt{2}}\ket{1}\right)\ket{0} + \left(\alpha s\ket{0}+\beta\frac{N-2s}{\sqrt{2}}\ket{1}\right)\ket{1}.
\end{eqnarray*}
If we now postselect on the last qubit being~1, the first qubit collapses to the state promised in the lemma.
\end{proof}

Our goal is to decide whether $|x|\geq N/2$ or not.
Consider the qubit of Eq.~(\ref{eqpostqubit}).
If $0<|x|<N/2$ then this qubit is strictly inside the first quadrant (i.e., both $\ket{0}$ and $\ket{1}$ have positive amplitude), and if $|x|\geq N/2$ then it is not.
In the first case, for some choice of $\alpha,\beta$ the qubit will be close to the state $\ket{+} = \frac{1}{\sqrt{2}}(\ket{0} + \ket{1})$,
while in the second case it will be far from $\ket{+}$ for \emph{every} choice of $\alpha,\beta$.
The algorithm tries out a number of $(\alpha,\beta)$-pairs in order to distinguish between these two cases.
Let $t$ be some positive integer (which we will later set to $\ceil{\log(2/\eps)}$ for our main algorithm).
Let 
$$
A = \{-\ceil{\log(N/t)},\ldots,-1,0,1,\ldots,\ceil{\log(N/t)}\},
$$ 
and for all $i\in A$ let $\ket{a_i}$ be the qubit of Eq.~(\ref{eqpostqubit}) with $\frac{\alpha}{\beta} = 2^i$.
Let 
$$
B = \{0,\dots,t-1\}\cup\{N/2-t+1,\ldots,N/2-1\}
$$ 
if $t\geq 2$, and $B = \emptyset$ otherwise. For all $i\in B$ let $\ket{b_i}$ be the qubit of Eq.~(\ref{eqpostqubit}) with $\frac{\alpha}{\beta} = \frac{N - 2i}{\sqrt{2}i}$. Note that $\ket{b_{|x|}}=\ket{+}$.

The intuition of the algorithm is that we are trying to eliminate from $A$ and $B$ all $i$ corresponding to states whose squared inner product with $\ket{+}$ is at most 1/2. If $|x|\geq N/2$ (i.e., $\MAJ_N(x)=1$) then we expect to eventually eliminate all~$i$, while if $|x|<N/2$ (i.e., $\MAJ_N(x)=0$) then for at least one~$i$, the squared inner product with $\ket{+}$ will be close to~1, and this~$i$ will probably not be eliminated by the process. 
We start with a procedure that tries to eliminate the elements of $A$:

\begin{lemma}\label{lemeliminateA}
For every integer $t\in\{1,\ldots,N/4\}$ there exists a postselection algorithm that uses $O(\log(N/t))$ queries to its input $x\in\01^N$ and distinguishes (with success probability $\geq 2/3$) the case $|x|\in\{t,\ldots,N/2-t\}$ from the case $|x|\geq N/2$.
\end{lemma}

\begin{proof}
The algorithm is as follows:
\begin{enumerate}
\item Initialize $k=1$ and $A_1 = A$. 
\item Repeat the following until $180\log(N/t)$ queries have been used (or until $A_k$ is empty):
\begin{enumerate}
\item For all $i\in A_k$:\\ 
create $5k$ copies of $\ket{a_i}$ and measure each in the $\ket{+},\ket{-}$ basis;\\ 
set $M_{k,i}=1$ if this resulted in a majority of $\ket{+}$ outcomes, and set $M_{k,i}=0$ otherwise.
\item Set $A_{k+1}=\{i\in A_k\mid M_{k,i}=1\}$.
Set $k$ to $k + 1$. 
\end{enumerate}
\item Output 0 if the final $A_k$ is nonempty, and output 1 otherwise. 
\end{enumerate}
Clearly the query complexity is $O(\log(N/t))$. We now analyze what happens in both cases.

{\bf Case 1: $|x|\in\{t,\ldots,N/2-t\}$.}
For these values of $|x|$, the ratio between $|x|$ and $N-2|x|$ lies between $t/N$ and $N/t$. Hence there exists an $i\in A$ such that $\ket{a_i}$ and $\ket{a_{i+1}}$ lie on opposite sides of $\ket{+}$. In the worst case, $\ket{+}$ lies exactly in the middle between $\ket{a_i}$ and $\ket{a_{i+1}}$, in which case $\inp{+}{a_i} = \inp{+}{a_{i+1}}$. In this case, $\ket{a_i} = \sqrt{\frac{1}{3}}\ket{0} + \sqrt{\frac{2}{3}}\ket{1}$, so $\inp{+}{a_i} = \frac{1 + \sqrt{2}}{\sqrt{6}} =: \lambda$. We will show that this~$i$ is likely to remain in all sets~$A_k$, in which case the algorithm outputs the correct answer~0.

Each iteration of step 2 will be called a ``trial''. Let $m$ be the number of the trial being executed when the algorithm stops (this~$m$ is a random variable).  The algorithm gives the correct output~0 iff $A_m$ is nonempty. First, by a Chernoff bound\footnote{\label{footnotechernoff}For $K$ coin flips $X_1,\ldots,X_K$, each taking value~1 with probability $p$, the probability that their sum $\sum_{i=1}^K X_i$ is at most $K(p-\eps)$, is upper bounded by $\exp(-2K\eps^2)$.  See for example~\cite[Appendix~A]{alon&spencer:probmethod}. We apply this here with $K=5k$, $p=\lambda^2\approx 0.97$, and $\eps=p-1/2$.} for every~$k$
$$
\Pr[M_{k,i} = 0] \leq \exp\left(-2\cdot 5k(\lambda^2 - 1/2)^2\right)\leq 2^{-(k+2)}.
$$ 
Now by the union bound, the error probability in this case is
$$
\Pr[A_m=\emptyset]\leq \Pr[i\notin A_m]=\Pr[\exists~k~s.t.~M_{k,i} = 0] \leq \sum_{k=1}^{\infty} 2^{-(k+2)}=\frac{1}{4}. 
$$

{\bf Case 2: $|x|\geq N/2$.}
We first show that the algorithm is likely to go through at least $\log N$ trials.
Since $|x|\geq N/2$, for all $i\in A$ we have $|\inp{+}{a_i}|^2\leq\frac{1}{2}$ and hence $\Pr[M_{k,i} = 1]\leq\frac{1}{2}$ for all~$k$. Therefore
$$
\E[|A_{k+1}|] = \sum_{i\in A} \prod_{\ell=1}^{k}\Pr[M_{\ell,i} = 1]
\leq \frac{|A|}{2^k} \leq \frac{\log(N/t)}{2^{k-1}}.
$$ 
Let $Q=\sum_{k=1}^{\log N} 5k|A_k|$ be the number of queries used in the first $\log N$ trials (with the number of queries set to~0 for the non-executed trials after the $m$th). Now: 
$$
\E[Q] \leq 5\log(N/t)\sum_{k=1}^{\log N} \frac{k}{2^{k-1}} \leq 20\log(N/t),
$$
where we used
$$
\sum_{k=1}^{\infty} \frac{k}{2^{k-1}}=\sum_{k=1}^\infty \sum_{\ell=k}^\infty \frac{1}{2^{\ell-1}}=4\sum_{k=1}^\infty 2^{-k}\sum_{\ell=1}^\infty \frac{1}{2^\ell}=4\sum_{k=1}^\infty 2^{-k}=4.
$$
By Markov's inequality 
$$
\Pr[Q\geq 180\log(N/t)]\leq\Pr[Q\geq 9\E[Q]]\leq\frac{1}{9}.
$$
So with probability at least $\frac{8}{9}$ we have $Q<180\log(N/t)$, meaning the algorithm executes at least $\log N$ trials before it terminates. In that case each element of $A$ has probability at most $1/2^{\log N}=1/N$ to survive $\log N$ trials. Hence, by the union bound
$$
\Pr[A_{2\log N + 1}\neq\emptyset]\leq \frac{|A|}{N} \leq\frac{1}{4},
$$
for $N$ sufficiently large.
Therefore the final error probability is at most $\frac{8}{9}\frac{1}{4} + \frac{1}{9}= \frac{1}{3}$ in this case. 
\end{proof}

Note that if we set $t=1$ in this lemma then we obtain an $O(\log N)$-query postselection algorithm that computes $\MAJ_N$ with error probability $\leq 1/3$ for all $x\neq 0^N$ (we can ensure $x\neq 0^N$ for instance by fixing the first two bits of $x$ to $01$, so then we would be effectively computing $\MAJ_{N-2}$). This improves upon the $O(\log(N)\log\log(N))$ algorithm mentioned in Section~\ref{ssecourresults}.  

We can reduce the error probability to any $\eps\in(0,1/2)$ by the standard method of running the algorithm $O(\log(1/\eps))$ times and taking the majority value among the outputs. This gives an $\eps$-error algorithm using $O(\log(N)\log(1/\eps))$ queries.  However, a slightly more efficient algorithm is possible if we set $t=\ceil{\log(2/\eps)}$ and separately handle the inputs with $|x|\notin\{t,\ldots,N/2-t\}$.

\begin{lemma}\label{lemeliminateB}
For every integer $t\in\{2,\ldots,N/4\}$ there exists a postselection algorithm that uses $O(t)$ queries to its input $x\in\01^N$ and distinguishes (with success probability $\geq 1-2^{-t}$) the case $|x|\in\{0,\ldots,t-1\}\cup\{N/2-t+1,\ldots,N/2-1\}$ from the case $|x|\geq N/2$.
\end{lemma}

\begin{proof}
The algorithm is as follows:
\begin{enumerate}
\item Initialize $B=\{0,\dots,t-1\}\cup\{N/2-t+1,\ldots,N/2-1\}$
\item Repeat the following $8t$ times (or until $B$ is empty):\\
take the first $i\in B$, create one copy of $\ket{b_i}$ and measure it in the $\ket{+},\ket{-}$ basis;\\  
if the outcome was $\ket{-}$ then remove $i$ from $B$.
\item Output 0 if the final $B$ is nonempty, and output 1 otherwise. 
\end{enumerate}
Clearly the query complexity is $O(t)$. We now analyze what happens in both cases.

{\bf Case 1: $|x|\in\{0,\ldots,t-1\}\cup\{N/2-t+1,\ldots,N/2-1\}$.}
Because $\ket{b_{|x|}}=\ket{+}$, the index $i=|x|$ will remain in $B$ with certainty.

{\bf Case 2: $|x|\geq N/2$.}
In this case, for all $i$ in the initial set $B$ we have $|\inp{+}{b_i}|^2\leq\frac{1}{2}$. Hence each measurement has probability at least 1/2 of producing outcome $\ket{-}$ and reducing the size of $B$ by~1. Since $B$ initially has $2t-1$ elements, it will only end up nonempty if there are fewer than $2t-1$ $\ket{-}$ outcomes among all $8t$ measurements. The probability of this event is upper bounded by the probability of $<2t-1$ ``heads'' among $K=8t$ fair coin flips. By the Chernoff bound (see footnote~\ref{footnotechernoff}, with $p=1/2$ and $\eps=1/4$), that probability is at most
$\exp(-2K(1/2-1/4)^2)=\exp(-t)\leq 2^{-t}$.
\end{proof}

To obtain our main algorithm we set $t=\ceil{\log(2/\eps)}$.
If $\eps\leq 2^{-\Omega(N)}$ then the trivial algorithm that queries all $N$ bits to determine Majority will be optimal up to a constant factor, so below we may assume $t\leq N/4$.
We now run the algorithm of Lemma~\ref{lemeliminateA} with error reduced to $\eps/2$, and the algorithm of Lemma~\ref{lemeliminateB} (with error $\leq 2^{-t}\leq \eps/2$), and we output 1 if both algorithms outputted~1. It is easy to see that this computes Majority with error probability $\leq\eps$ on every input. This proves:

\begin{theorem}\label{th:majalgo}
For every $\eps\in(2^{-N},1/2)$ there exists a postselection algorithm that computes $\MAJ_N$ using $O\left(\log(N/\log(1/\eps))\cdot \log(1/\eps)\right)$ queries with error probability $\leq\eps$.
\end{theorem}

The latter algorithm is asymptotically better than the earlier $O(\log(N)\log(1/\eps))$ algorithm if $\eps$ is slightly bigger than $2^{-N}$.  For example, if $\eps=2^{-N/\log N}$ then the earlier algorithm has query complexity $O(N)$ while Theorem~\ref{th:majalgo} gives $O(N\log\log(N) / \log(N))=o(N)$.

Sherstov~\cite[Theorem~1.7]{sherstov:intersection} proved an $\Omega(\log(N/\log(1/\eps))\cdot \log(1/\eps))$ lower bound on the degree of $\eps$-approximating rational functions for $\MAJ_N$, for all $\eps\in(2^{-N},1/2)$.  Together with our Theorem~\ref{thrdeg2postq}, this shows that the algorithm of Theorem~\ref{th:majalgo} has optimal query complexity up to a constant factor.

\section{Deriving Newman's Theorem}\label{secnewman}

We now use the postselection algorithm for Majority to derive a good, low-degree rational approximation for the sign-function:

\begin{theorem}
For every $d$ there exists a degree-$d$ rational function that $\eps$-approximates the sign-function $\sign(z)$ on $[-1,-\eps]\cup[\eps,1]$ for $\eps = 2^{-\Omega(\sqrt{d})}$ (and which lies in $[-1,1]$ for all $z\in[-1,1]$).
\end{theorem}

\begin{proof}
Set $\eps = 2^{-\Omega(\sqrt{d})}$ with a sufficiently small constant in the $\Omega(\cdot)$, and $N = \ceil{\frac{2}{\eps}}$. Consider the algorithm described after Lemma~\ref{lemeliminateA} with $t=1$ and error reduced to $\eps/2$. It provides two $N$-variate multilinear polynomials $P$ and $Q$, each of degree $d=O(\log(N)\log(1/\eps))=O(\log(1/\eps)^2)$, such that for all $x\in\01^N$,
$$
\left|\frac{P(x)}{Q(x)} - \MAJ_N(x)\right|\leq\frac{\eps}{2}.
$$
Note that $P$ can be written as $\sum_{j} c_j(\sum_i x_i)^j$, as can $Q$, because the amplitudes of the states $\ket{a_i}$ and $\ket{b_i}$ in the proof of Theorem~\ref{th:majalgo} are functions of~$|x|=\sum_i x_i$. To convert $P$ to a univariate polynomial $p$, replace $\sum_i x_i$ with real variable~$z$ to obtain $p(z) = \sum_j c_jz^j$. Similarly convert $Q(x)$ to $q(z)$. Let $\maj_N$ represent the univariate version of $\MAJ_N$: $\maj_N$ returns~0 on input $x\in[0,\dots,\frac{N}{2})$ and returns~1 on $x\in[N/2,\dots,N]$. We now have:
$$
\left|\frac{p(x)}{q(x)} - \maj_N(x)\right|\leq\frac{\eps}{2}
$$ 
for $x\in\{0,\dots,N\}$. Crucially, this inequality also holds for real values $z\in [1,N/2 - 1]\cup [N/2,N]$. This is because the analysis of the algorithm described after Lemma~\ref{lemeliminateA} (with $t=1$ and error reduced to $\eps/2$) still works when we replace the integer $|x|$ with real value~$z$. Since $\sign(z) = 2\maj_N(\frac{N(z+1)}{2}) - 1$, we have
$$
\left|\frac{2p\left(\frac{N(z+1)}{2}\right)- q\left(\frac{N(z+1)}{2}\right)}{q\left(\frac{N(z+1)}{2}\right)} - \sign(z) \right|\leq\eps
$$
for all $z\in [-1,-\frac{2}{N}] \cup [0,1]$. Since $N = \ceil{\frac{2}{\eps}}$, we have the desired approximation on $[-1,-\eps]\cup[\eps,1]$.
\end{proof}

It is easy to see that multiplying the above rational function by $z$ gives an approximation of the absolute-value function $|z|$ on the whole interval $z\in [-1,1]$. Thus we have reproved Newman's Theorem in a new, quantum-based way:

\begin{corollary}[Newman]
For every integer $d\geq 1$ there exists a degree-$d$ rational function that approximates $|z|$ on $[-1,1]$ with error $\leq 2^{-\Omega(\sqrt{d})}$.
\end{corollary}



\section{Open questions}

We mention a few open questions.
First, we have very few techniques for quantum algorithms with post\-selection.  Aaronson's techniques from~\cite{aaronson:pp} (and our variations thereof) is the main technique we know that makes non-trivial use of the power of postselection.  What other algorithmic tricks can we play using postselection?
Using the equivalence between postselection algorithms and rational degree, we can try to obtain new algorithms from known rational approximations.
Very tight bounds are known for the rational degree of approximations of the univariate exponential functions $\exp(x)$ and $\exp(-x)$~\cite[Sections~4.4 and~4.5]{petrushev&popov:rational}.  In particular, rational degree $d$ is necessary and sufficient to achieve approximation-error $\exp(-\Theta(d))$ for the function $\exp(-x)$ on the interval $[0,\infty)$. This implies the following for postselection algorithms. Consider the real-valued $n$-bit function $f:\01^n\rightarrow\mathbb{R}$ defined by $f(x)=\exp(-|x|)$.  Then for every integer~$d>0$ there exists a quantum algorithm with postselection, that makes $O(d)$ queries to its input~$x\in\01^n$, and whose acceptance probability is within $\exp(-d)$ of $f(x)$. Can we use such a postselection algorithm to compute something useful?

Second, we showed here how a classical but basic theorem in rational approximation theory (Newman's theorem) could be reproved based on efficient quantum algorithms with postselection.  Is it possible to prove \emph{new} results in rational approximation theory using such algorithms?

Finally, the following is a long-standing open question attributed to Fortnow by Nisan and Szegedy~\cite[p.~312]{nisan&szegedy:degree}: is there a polynomial relation between the \emph{exact} rational degree of a Boolean function $f:\01^N\rightarrow\01$ and its usual polynomial degree? It is known that exact and bounded-error quantum query complexity and exact and bounded-error polynomial degree are all polynomially close to each other~\cite{buhrman&wolf:dectreesurvey}, so rephrased in our framework Fortnow's question is equivalent to the following: can we efficiently simulate an \emph{exact} quantum algorithm with postselection by a bounded-error quantum algorithm without postselection?\footnote{Note that we are asking about \emph{exact} rational degree here; for $\eps$-approximate rational degree the Majority function gives an example of an exponential gap between rational degree and the usual polynomial degree.}  
We hope this more algorithmic perspective will help answer his question.

\bigskip

\noindent
{\bf Acknowledgment.} We thank Andr\'e Chailloux for helpful discussions, and Sushant Sachdeva for asking us about rational approximations of exponential functions. We also thank the anonymous QIC referees for many helpful comments.

\bibliographystyle{alpha}

\newcommand{\etalchar}[1]{$^{#1}$}

\end{document}